\documentclass[reqno]{amsart}


\usepackage{amsmath}
\usepackage{amssymb}
\usepackage{amsthm}
\usepackage{amsfonts}
\usepackage{enumerate}
\usepackage{mathtools}
\usepackage{bm}
\usepackage{ulem}
\usepackage{fancyhdr}
\usepackage{hyperref}
\usepackage{mdframed}
\usepackage{tcolorbox}
\usepackage{graphicx,lipsum,wrapfig}
\usepackage[font={tiny}]{caption}
\usepackage[font={scriptsize}]{caption}

\usepackage{xspace}
\usepackage{subcaption}
\usepackage{todonotes}
\usepackage{mathtools}
\usepackage{xcolor}

\theoremstyle{plain}
\newtheorem{thm}{Theorem}[section]

\newtheorem{lemma}[thm]{Lemma}
\newtheorem{proposition}[thm]{Proposition}
\newtheorem*{theorem*}{Theorem}
\newtheorem{remark}[thm]{Remark}
\newtheorem{definition}[thm]{Definition}
\newtheorem{problem}[thm]{Problem}

\numberwithin{equation}{section}

\renewcommand{\emph}[1]{\textit{#1}}

\usepackage{times}
\usepackage{booktabs}
\usepackage{dsfont}
\usepackage[T1]{fontenc}
\usepackage[utf8]{inputenc}
\usepackage{eso-pic}

\newcommand{\IG}[1]{\textcolor{green}{[IG: #1]}}
\newcommand{\LL}[1]{\textcolor{red}{[LL: #1]}}

\renewcommand{\path}{\omega}
\newcommand{\x}{\textbf{{x}}} 
\newcommand{\w}{\textbf{{w}}}
\newcommand{\pathX}{\path_{x}}
\newcommand{\pathXbold}{\boldsymbol{\path}_{x}}
\newcommand{\PathX}{\Omega_{x}}

\newcommand{\policy}{\sigma}
\newcommand{\Setofpolicies}{\Sigma}
\newcommand{\adversary}{\xi}
\newcommand{\Adversary}{\Xi}

\newcommand{\pathrmdp}{\omega}
\newcommand{\pathrmdpbold}{\boldsymbol{\omega}}
\newcommand{\Pathrmdp}{\Omega}
\newcommand{\PRMDP}{\M^\varphi}

\newcommand{\last}{\mathrm{last}}
\newcommand{\reals}{\mathbb{R}}
\newcommand{\naturals}{\mathbb{N}}
\newcommand{\M}{\mathcal{M}}
\newcommand{\I}{\mathcal{I}}

\newcommand{\D}{\mathcal{D}}
\newcommand{\A}{\mathcal{A}}
\newcommand{\LTLf}{LTL$_f$\xspace}

\newcommand{\Globally}{\mathcal{G}}
\newcommand{\Eventually}{\mathcal{F}}
\newcommand{\Until}{\mathcal{U}}
\newcommand{\Next}{\mathcal{X}}

\newcommand{\Prop}{\mathfrak{p}}
\newcommand{\safe}{\mathrm{safe}}

\begin{document}

\AddToShipoutPictureBG*{%
  \AtPageUpperLeft{%
    \hspace{16.5cm}%
    \raisebox{-1.5cm}{%
      \makebox[0pt][r]{To appear in \textit{Learning for Dynamics and Control Conference} (\textit{L4DC}), July 2024.}}}}

\title[Data-Driven Strategy Synthesis]{Data-Driven Strategy Synthesis for Stochastic Systems with Unknown Nonlinear Disturbances}

\author[I.Gracia]{Ibon Gracia}

\author[D. Boskos]{Dimitris Boskos}

\author[L. Laurenti]{Luca Laurenti}

\author[M. Lahijanian]{Morteza Lahijanian}

\thanks{The authors are  with the Smead
Department of Aerospace Engineering Sciences, CU Boulder, Boulder, CO, \tt{\{ibon.gracia, morteza.lahijanian\}@colorado.edu}, and 
the Delft Center for Systems and Control, 
Faculty of Mechanical, Maritime and Materials Engineering, Delft University of Technology \tt{\{d.boskos,l.laurenti\}@tudelft.nl}.}

\begin{abstract}
In this paper, we introduce a data-driven framework for synthesis of provably-correct controllers for general nonlinear switched systems under complex specifications. The focus is on systems with \emph{unknown} disturbances whose effects on the dynamics of the system is nonlinear. 
The specifications are assumed to be given as \emph{linear temporal logic over finite traces} (\LTLf) formulas.
Starting from observations of either the disturbance or the state of the system, we first learn an \emph{ambiguity set} that contains the unknown distribution of the disturbances with a user-defined confidence. Next, we construct a \emph{robust Markov decision process} (RMDP) as a finite abstraction of the system.  By composing the RMDP with the automaton obtained from the \LTLf formula and performing optimal robust value iteration on the composed RMDP, we synthesize a strategy that yields a high probability that the uncertain system satisfies the specifications.  Our empirical evaluations on systems with a wide variety of disturbances show that the strategies synthesized with our approach lead to high satisfaction probabilities and validate the theoretical guarantees.%
\end{abstract}

\maketitle


\section{Introduction}
\label{sec:introduction}

Stochastic dynamical systems are powerful mathematical models commonly employed to describe complex systems in many applications, ranging from autonomous systems \cite{kumar2015stochastic} to biological systems \cite{anderson2011continuous}. Many of these applications are \emph{safety-critical}, and formal correctness guarantees are required to avoid fatal or costly accidents. Another common feature of these applications is that they are characterized by nonlinear dynamics and that the characteristics of the noise affecting the system dynamics are often uncertain and need to be estimated from data \cite{rahimian2019distributionally}. That is, the system is both characterized by \emph{aleatoric} uncertainty (i.e., the uncertainty intrinsic in the system dynamics) and by \emph{epistemic} uncertainty (i.e., the uncertainty due to the lack of knowledge of the system dynamics). 
A challenging research question that still remains open is \emph{how to provide formal correctness guarantees for nonlinear stochastic systems whose noise distribution is unknown?} which is the focus of this work.

Formal guarantees of correctness for nonlinear stochastic systems are generally provided either via formal abstractions into a finite transition model, typically a variant of a Markov model \cite{lahijanian2015formal,cauchi2019efficiency,lavaei2022automated}, or via stochastic barrier functions \cite{jagtap2020formal,mazouz2022safety,santoyo2021barrier,lechner2022stability}. 
However, all those works assume a known model and  known disturbances, i.e., no \emph{epistemic uncertainty}.
To tackle this issue and account for epistemic uncertainty, data-driven formal approaches have been recently developed to estimate model dynamics from data \cite{martin2023regret,jackson2021strategy, schon2023bayesian}. These works generally employ concentration inequalities to estimate the system dynamics from noisy measurements of the system and then rely on existing formal methods for verification and control of the estimated model. However, these approaches generally assume that the noise distribution of the system is known and only the vector field of the dynamics needs to be learnt. In addition, they often restrict themselves to linear disturbances. The complementary case where the noise distribution is unknown and needs to be estimated from data is considered in \cite{badings2023robust,mathiesen2023inner, gracia2023distributionally}, which however, are specific for linear systems or additive noise.

In this paper, we propose a formal data-driven approach for controller synthesis of nonlinear stochastic systems with unknown disturbances against linear temporal logic specifications over finite traces (LTLf) \cite{de2013linear}. Our approach combines tools from optimal transport theory with abstractions methods. In particular, given a dataset, we first exploit results from \cite{fournier2022convergence} to build an \emph{ambiguity set} for the disturbance distribution, that is, a set of distributions that contains the unknown law of the disturbance with high confidence. Then, we abstract the resulting uncertain model into a robust Markov decision process (RMDP) \cite{nilim2005robust}, which formally accounts for both aleatoric and epistemic uncertainty in the system. The resulting abstraction is employed to synthesize a controller that maximizes the probability that the system satisfies an \LTLf property and is robust against the uncertainties, thus guaranteeing correctness. 
The efficacy of our framework is illustrated on various benchmarks showing how our approach can successfully synthesize high-confidence formal controllers for uncertain nonlinear stochastic systems from data.

The key contributions of this paper are:
\vspace{-2mm}
\begin{itemize}
    \item a novel strategy synthesis framework for nonlinear switched stochastic systems with general (e.g., non-additive) disturbances of unknown distributions under \LTLf specifications,
    \item an efficient procedure to obtain a finite and sound abstraction of said class of systems,
    \item probabilistic satisfaction guarantees whose confidence is independent of the abstraction size, 
    \item extensive case studies and empirical evaluations involving five different systems, several choices of hyperparameters, and comparisons with the state-of-the-art approaches.
\end{itemize}
\vspace{-1mm}
\noindent
\subsection*{Basic Notation}
\label{sec:basic_notation}

The set of non-negative integers is defined as $\naturals_0 = \naturals\cup\{0\}$. For a set $X$, $|X|$ denotes its cardinality. Given $X \subseteq \mathbb{R}^n$ we denote by $\mathds{1}_{X}$ its indicator function, that is $\mathds{1}_{X}(x) = 1$ if $x\in X$ and $0$ otherwise. We let $\mathcal{B}(X)$ be the Borel $\sigma$-algebra of $X$ and $\mathcal{P}(X)$ denote the set of probability distributions on $(X, \mathcal{B}(X))$. For $P\in\mathcal{P}(X)$ and $x\in X$, $P({x})$ denotes the probability of the (singleton) event $\{x\}$. Given a continuous cost $c:X\times X\rightarrow \mathbb{R}_{\ge 0 }$, we denote by $\mathcal{P}_c(X)$ the set of distributions $P\in\mathcal{P}(X)$ with $\int_X c(x,x') P(dx) < \infty$ for some $x'\in X$. The optimal transport discrepancy between  $P,P'\in \mathcal{P}_c(X)$ is defined as
$
    \mathcal{T}^{c}(P,P') = \inf_{\pi\in\Pi(P,P')}\int_{X\times X} c(x,x')\pi(dx,dx'),
$ 
where $\Pi(P,P')$ is the set of distributions on $\mathcal{P}(X\times X)$ with marginals $P$ and $P'$. For an $l$-norm $\|\cdot\|_l$, with $l \in[1,\infty]$, and $s\ge 1$, if $c(x,x') = \|x - x'\|_l^s$ for all $x,x'\in X$, we use the notation $\mathcal{T}_s^{(l)}(P,P')$. Furthermore, if $X\subseteq \mathbb{R}^d$, then $\mathcal{W}_s(P,P') :=\big( \mathcal{T}_s^{(l)}(P,P') \big)^{1/s}$ is the $s$-Wasserstein distance between $P$ and $P'$.
We let $\delta_{x}\in\mathcal{P}(X)$ be the Dirac measure located at $x\in X$. We use bold symbols, e.g., $\x$, to denote random variables.
\section{Problem Formulation}
\label{sec:problem_formulation}

In this work, we focus on discrete-time switched stochastic processes described by 
%
\begin{align}
\label{eq:system}
    \x_{k+1} = f_{a_k}(\x_k,\w_k),
\end{align}
where
$\x_k\in \mathbb{R}^n$ is the state and $a_k\in A := \{a_1,\dots,a_{|A|}\}$ is the control (mode or action) at time $k$. Noise term $\w_k \in W$ is an i.i.d. random variable  taking values in a bounded set $W\subset \mathbb{R}^d$ with probability distribution
$P_w$ identical at each time step $k$. Distribution $P_w$, however, is assumed to be \emph{unknown}.  For each $a\in A$, $f_a: \reals^n \times W \to \reals^n$ is the vector field, which is assumed to be
Lipschitz continuous in $w$ with Lipschitz constant $L_w(x,a)\geq 0$, i.e., 
$\|f_a(x,w) - f_a(x,w')\| \le L_w(x,a) \|w - w'\|$ for all $w,w'\in W.$
Consequently, Process \eqref{eq:system} describes a general model of nonlinear switched stochastic processes, which encompasses many models commonly used in practice. 

For $x_0,\ldots,x_K \in \mathbb{R}^n$ and $K \in \naturals_0$, $\pathX =x_0 \xrightarrow{a_0} x_1 \xrightarrow{a_1}  \ldots \xrightarrow{a_{K-1}} x_K$ represents a finite \emph{trajectory} 
of Process~\eqref{eq:system}, and $\PathX$ is the set of all finite trajectories (of any length).
We denote by $\pathX(k)$ 
the state of $\pathX$ at time $k\in\{0,\dots, K\}$. 
A \emph{switching strategy} $\policy_x: \PathX \to A$ is a function that assigns an action $a \in A$ to each finite trajectory. 
For $x\in\reals^n$, $a\in A$ and $B\in\mathcal{B}(\reals^n)$, the \emph{transition kernel}
of Process~\eqref{eq:system} is
$T^a(B\mid x) = \int_W \mathds{1}_{B}(f_a(x,v))P_w(dv) $.
Given a strategy $\policy_x$ and an initial condition $x_0\in\reals^n$, the transition kernel uniquely defines a probability measure $P_{x_0}^{\policy_x}$ over the trajectories of Process \eqref{eq:system} 
\cite{bertsekas1996stochastic} such that $P_{x_0}^{\policy_x}[\pathXbold(k) \in X]$
represents the probability that, under strategy $\policy_x$, $\x_k $ is in set $X\subseteq \reals^n$ starting from $x_0$.
%
%

%
%

%

Given a bounded set $X\subset \mathbb{R}^n$, we are interested in the temporal properties of Process~\eqref{eq:system}
w.r.t. a
set of regions of interest $R := \{r_1,\dots,r_{|R|-1}, r_u\}$ with $r_i\subseteq X$ for all $i\leq |R|-1$, and $r_u = \reals^n\setminus X$. To each region $r_i\in R$, we associate an atomic proposition $\Prop_i$, and say that $\Prop_i \equiv \top$ (i.e., $\Prop_i$ is \textit{true}) at $x \in \reals^n$ iff $x\in r_i$. Denote by $AP := \{\Prop_i,\dots,\Prop_{|R|-1}, \Prop_u\}$ the set of all atomic propositions and define the labeling function $L:\reals^n\rightarrow 2^{AP}$ as the function that maps each state $x\in \reals^n$ to the set of atomic propositions that are true at that state, i.e., $\Prop_i\in L(x)$ iff $x\in r_i$. Then, to each trajectory $\pathX= x_0 \xrightarrow{a_0} x_1 
\xrightarrow{a_1}  \ldots \xrightarrow{a_{K-1}} x_K$,
we can associate the (observation) \emph{trace} 
$\rho = \rho_0\rho_1\dots\rho_K, $
where $\rho_k := L(x_k)$ for all $k\in\{0,\dots,K\}$. 

To specify temporal properties of Process~\eqref{eq:system}, we use \LTLf \cite{de2013linear}, which
is a formal language commonly employed to express finite behaviors.
%
    An \LTLf formula $\varphi$ is defined from a set of atomic propositions $AP$ and is closed under the Boolean connectives ``negation'' ($\neg$) and ``conjunction'' ($\land$), and the temporal operators ``until'' ($\Until$) and ``next'' ($\Next$):
    \begin{align*}
        \varphi := \top \mid \Prop \mid \neg\varphi \mid \varphi_1\land\varphi_2 \mid \Next\varphi \mid \varphi_1\Until\varphi_2 
    \end{align*}
    where $\Prop\in AP$ and $\varphi_1,\varphi_2$ are \LTLf formulas themselves.
%
The temporal operators ``eventually'', $\Eventually$, and ``globally'', $\Globally$, are defined as  $\Eventually \varphi := \top\Until\varphi$ and $\Globally \varphi := \neg\Eventually(\neg\varphi)$.
The semantics (interpretation) of \LTLf are defined over finite traces \cite{de2013linear}. We say a trajectory $\pathX$ satisfies a formula $\varphi$, denoted by $\pathX \models \varphi$, if a prefix of its trace $\rho$ satisfies the formula.

We are interested in synthesizing a strategy that maximizes the probability that Process~\eqref{eq:system} satisfies a given \LTLf formula, while accounting for the uncertainty coming from the fact that $P_w$ is unknown and can only be estimated from data. A formal statement of the problem is as follows.

\begin{problem}
    \label{prob:Syntesis}
    %
    Consider the switched stochastic Process~\eqref{eq:system}, a set $\{\w^{(i)}\}_{i=1}^N$ of $N$ i.i.d. samples of $P_w$, a bounded set $X\subset\mathbb{R}^n$, and an \LTLf formula $\varphi$ defined over the regions of interest $R$. Then, given a confidence level $1-\beta \in (0,1)$, synthesize a switching strategy $\sigma_x$ such that, for every initial state $x_0\in X$, with confidence at least $1-\beta$, $\sigma_x$ maximizes the probability that the paths $\pathX \in \PathX$ satisfy $\varphi$ while remaining in $X$, i.e.,
      $\sigma_x \in \arg \max_{\sigma_x}  P_{x_0}^{\sigma_x}[\pathX \models \varphi \land \Globally \neg \Prop_u]$.
\end{problem}
Note that in Problem \ref{prob:Syntesis} the noise distribution $P_w$ is unknown. Consequently, $P_w$ must be estimated from data. This inevitably leads to a confidence probability on the resulting strategy, which irrespective of the method of quantification, depends on the data. Problem~\ref{prob:Syntesis} requires this confidence to be at least $1-\beta$. We should also remark that we assume that a set of $N$ samples from the noise is given. 
%
This may seem like a strong assumption as, in general, only observations $\{\mathbf y^{(i)}\}_{i=1}^{N+1}$ of the state can be collected. However, in several cases, e.g., when successive full-state observations $\mathbf y^{(i)}=\mathbf x^{(i)}$ are available, one can use the system's dynamics $\mathbf y^{(i+1)}=f_{a^{(i)}}(\mathbf y^{(i)},\mathbf w^{(i)})$ and knowledge of its actions $a^{(i)}$ to obtain i.i.d. samples $\{\mathbf w^{(i)}\}_{i=1}^{N}$ of the noise distribution $P_w$.
For instance, this is possible when  $f$ is linear in $w$ or, more generally, if the data are collected where $f_{a^{(i)}}(\mathbf y^{(i)},\cdot):W\longrightarrow\mathbb{R}^n$ is injective. Such regions can be easily identified in many non-linear systems commonly employed in practice, as shown in the examples in Section~\ref{sec:case_studies}.

\paragraph{Overview of the approach} {
Since $P_w$ is unknown and the state space of the process is uncountable,  obtaining the exact solution to Problem \ref{prob:Syntesis} is generally infeasible. Hence, we follow an abstraction-based approach instead and synthesize an optimal strategy that is robust against the uncertainty in estimating $P_w$ from data.} First, in Section \ref{sec:learning},
we use the noise samples to learn an ambiguity set of probability distributions that contains $P_w$ with confidence $1-\beta$. Next, in Section \ref{sec:imdp_abstraction}, we obtain an interval MDP (IMDP) abstraction of Process~\eqref{eq:system} by considering that $\w$ is distributed according to a fixed distribution supported on the samples.
In Section~\ref{sec:rmdp_abstraction}, to account for
all the distributions in the ambiguity set, we expand the set of transition probabilities of the IMDP, obtaining a robust MDP. Then, in Section~\ref{sec:strategy_synthesis}, we synthesize a (robust) optimal strategy for the robust MDP via robust dynamic programming. Finally, we refine this strategy to the original system, obtaining a guaranteed lower bound on the satisfaction of the \LTLf specification for Process \eqref{eq:system}.

\section{Preliminaries on Robust Markov Decision Processes}
\label{sec:robust_MDPs}

A robust MDP (RMDP) is a generalization of a Markov decision process in which the transition probability distributions are uncertain and belong to an ambiguity set \cite{nilim2005robust,wiesemann2013robust}.
\begin{definition}[Robust MDP]
\label{def:robust_mdp}
    A labelled robust Markov decision process (RMDP) $\M$ is a tuple $\M = (Q,A,\Gamma,q_0,AP,L)$, where
    \begin{itemize}
        \setlength\itemsep{0.1mm}
    	\item $Q$ is a finite set of states
    	\item $A$ is a finite set of states and actions, and $A(q)$ denotes the actions available at state $q\in Q$,
        \item $\Gamma = \{\Gamma_{q,a} : {q\in Q, a\in A(q)}\}$, where 
        $\Gamma_{q,a} \subseteq \mathcal{P}(Q) $ is the set of
        transition probability distributions for state-action pair $(q,a) \in Q\times A(q)$,
        \footnote{Notice that for two distinct state-action pairs, the respective sets of transition probability distributions of $\M$ are independent. This is the well-known \textit{rectangular} property of robust MDPs \cite{nilim2005robust,wiesemann2013robust}.}
        \item $q_0 \in Q$ is initial state,
        \item $AP$ is a finite set of atomic propositions
        \item $L:Q\longrightarrow 2^{AP}$ is the labeling function.
    \end{itemize}
\end{definition}
A finite \emph{path} of RMDP $\M$ is a sequence of states $\pathrmdp = q_0 \xrightarrow{a_0} q_1 \xrightarrow{a_1} \ldots 
\xrightarrow{a_{K-1}} q_K$ such that $a_k \in A(q_k)$ and there 
exists $\gamma\in\Gamma_{q_k,a_k}$ with $ \gamma(q_{k+1})> 0$ for all $k \in \{0, \dots, K-1\}$. We denote the set of all paths of finite length by $\Pathrmdp$. Given a path $\pathrmdp\in\Pathrmdp$, we let $\pathrmdp(k) = q_k$ be the state of $\pathrmdp$ at time $k \in \{0,\dots, K\}$, and $\last(\pathrmdp)$ be its last state. 
A labelled interval-valued MDP (IMDP) $\I$, also called bounded-parameter MDP (BMDP) \cite{givan2000bounded}, is an RMDP $\I = (Q, A,\Gamma,q_0,AP,L)$ where the transition probabilities lie within independent intervals. That is, 
$\Gamma_{q,a} = \{ \gamma\in\mathcal{P}(Q): \underline P(q,a,q')\leq \gamma(q')\leq \overline P(q,a,q')\:\text{for all}\: q,q'\in Q,a\in A(q)\}.$
%

%
Given a finite path of an RMDP or IMDP, a \emph{strategy}  $\policy$ chooses the next action. 
 Formally, $\policy: \Pathrmdp \rightarrow A$ is a function that assigns an action $a \in A$ to each finite path $\pathrmdp\in\Pathrmdp$ of $\M$. 
Any given strategy $\policy\in\Setofpolicies$ restricts the set of feasible transition probability distributions of the RMDP, thus yielding a set of Markov chains, or robust Markov chain. To further reduce this set to a single Markov chain, we define the adversary 
function~\cite{givan2000bounded} 
that chooses the distribution of the next state.
Formally, an \emph{adversary} is a function $\adversary: \Pathrmdp \times A \rightarrow \mathcal{P}(Q)$ that maps each finite path $\pathrmdp \in \Pathrmdp$ and action $a \in A(\last(\pathrmdp))$ to an admissible distribution $\gamma\in\Gamma_{q,a}$. The set of all adversaries is denoted by $\Adversary$.
 Given an initial condition $q_0\in Q$, a strategy $\policy\in\Setofpolicies$ and an  adversary $\xi\in\Xi$, the RMDP collapses to a Markov chain with a unique probability distribution $Pr_\xi^{q_0,\policy}$ on its paths.
 
\section{Robust MDP Abstraction}

Our approach is based on finite abstraction of Process~\eqref{eq:system} to an RMDP via ambiguity set learning. First, in Section \ref{sec:learning} we show how we can build an ambiguity set that contains $P_w$ with high confidence, and then, in the rest of the Section, we show how to build an RMDP abstraction \newline $\M = (Q,A,\Gamma,AP, q_0, L)$ that formally accounts for this uncertainty.

\subsection{Learning the Unknown Probability Distribution}
\label{sec:learning}
In this section, we make use of the samples of $\w$ 
to obtain a formal representation of the unknown distribution $P_w$. Specifically, we learn an ambiguity set, i.e., a set of probability distributions that is guaranteed to contain $P_w$ with high confidence. In particular, Lemma \ref{lemma:ambiguity_set} below allows us to build an ambiguity set of distributions centered on the empirical distribution generated by the samples.
%
\begin{lemma}[Ambiguity set]
\label{lemma:ambiguity_set}
    Given $s\in \mathbb N$, $l\in[1,\infty]$, and $N$ i.i.d. samples $\w^{(1)},\ldots,\w^{(N)}$ from $P_{w}$, define their \emph{empirical distribution} as $\widehat P_w :=  \frac{1}{N}\sum_{i=1}^N\delta_{\w^{(i)}}$. Further, let $\phi$ be the diameter of $W$ in the $\infty$-norm, $\beta\in (0,1)$, and
    $g(N,\phi,d,s,l)$ be the upper bound on the expected transport cost  $\mathbb{E}[\mathcal{T}_s^{(l)}(P_w,\widehat P_w)]$ between $P_w$ and $\widehat P_w$ given in \cite{fournier2022convergence}. 
    Then the Wasserstein ball
    \begin{align}
    \label{eq:wass_ball}
        \mathcal{D} := \{ P \in \mathcal{P}(W) : \mathcal{W}_s(P, 
        \widehat P_w) \le \varepsilon (N,\beta) \},
    \end{align}
    where
    %
    %
    \begin{align*}
        \varepsilon(N,\beta) := g(N,\phi,d,s,l)^{\frac{1}{s}} 
    + \phi\sqrt{d}(2\log{1/\beta})^{\frac{1}{2s}}N^{-\frac{1}{2s}},
    \end{align*}
    contains $P_w$ with confidence $1-\beta$.
\end{lemma}
\begin{proof}
%
%
\textcolor{black}{
In analogy to the proof of \cite[Proposition 24 of arXiv version]{boskos2023high}, 
we leverage the fact that $\mathcal{W}_s(P_w,\widehat P_w)$ concentrates around its mean \cite{boissard2014mean}:}
%
\begin{align*}
    \mathbb{P}\big( \mathcal{W}_s^{(l)}(P_w,\widehat P_w) \ge \mathbb{E}[\mathcal{W}_s^{(l)}(P_w,\widehat P_w)] + \tau \big) \le e^{-N\tau^{2s}/(2\phi^{2s})},\:\forall \tau \ge 0.
\end{align*}

%
By the definition of the $s$-Wasserstein distance and Jensen's inequality,
\begin{align*}
    \mathbb{E}[\mathcal{W}_s^{(l)}(P_w,\widehat P_w)] := \mathbb{E}[\big(\mathcal{T}_s^{(l)}(P_w,\widehat P_w)\big)^{1/s}] \le \mathbb{E}[\mathcal{T}_s^{(l)}(P_w,\widehat P_w)]^{1/s},
\end{align*}
and therefore
\begin{align*}
    \mathbb{P}\big( \mathcal{W}_s^{(l)}(P_w,\widehat P_w) \ge \mathbb{E}[\mathcal{T}_s^{(l)}(P_w,\widehat P_w)]^{1/s} + \tau \big) \le e^{-N\tau^{2s}/(2\phi^{2s})},\: \forall \tau \ge 0.
\end{align*}
Selecting $\tau$ such that  $\beta= e^{-N\tau^{2s}/(2\phi^{2s})}$ and setting $\varepsilon := \mathbb{E}[\mathcal{T}_s^{(l)}(P_w,\widehat P_w)]^{1/s} + \tau$, we arrive at the desired result.
\end{proof}

{It is important to stress that in all cases $g\to 0$ as $N$ increases. Thus, the radius of the ambiguity set $\mathcal{D}$ shrinks with more data. }

\subsection{States, Actions and Labeling Function of the Abstraction}
\label{sec:states_actions_label}

The state-space of the abstraction is obtained as follows. First, we partition set $X$ into the set of non-empty, non-overlapping regions $Q_{\safe} := \{q_1,\dots,q_{|Q_{\safe}|}\}$, i.e., $q\cap q' = \emptyset$ for all $q\neq q'\in Q_{\safe}$ and $\cup_{q\in Q_{\safe}}q = X$. Furthermore, this partition must respect the regions of interest, i.e., for every $r \in R$ and $q\in Q_{\safe}$, 
$q\cap r \in \{\emptyset, q\}$.
The state space is then defined by associating each state of the abstraction with a region in $Q_{\safe}$ and the additional ``unsafe" region $q_u = \reals^n \setminus X$, i.e., $Q := Q_{\safe}\cup\{q_u\}$. With a small abuse of notation, we let $q\in Q$ denote at the same time a state of the abstraction and a region in $\reals^n$. We also let the actions of the abstraction $A$ be the same as those of Process \eqref{eq:system}, with $A(q) := A$ for all $q \in Q$.
Furthermore, with an additional abuse of notation, we define the abstraction labeling function as $L:Q\rightarrow 2^{AP}$ such that $L(q) := L(x)$ if $x\in q$ for all $q\in Q$. 
Note that we do not fix the initial state, since our approach yields results for every $q_0 \in Q$.

\subsection{Transition Probability Distributions of the Abstraction}

\subsubsection{Accounting for the discretization error}
\label{sec:imdp_abstraction}

To embed the discretization error into the abstraction, we first abstract Process~\eqref{eq:system} to an IMDP $\I := (Q, A, \widehat\Gamma,q_0, AP, L)$, namely the ``empirical IMDP'', under the assumption that $\w_k$ is distributed according to the empirical distribution $\widehat P_w$. Components $Q$, $A$, $q_0$, $AP$, and $L$ of $\I$ are as defined above. The following proposition defines the transition probability intervals $\widehat\Gamma$ of $\I$.

\begin{proposition}
    Let $T_{\widehat P_w}^u$ be the transition kernel of Process~\eqref{eq:system} when $\w_t\sim\widehat P_w$, i.e.,
    %
    %
    for $B\in\mathcal{B}(\reals^n)$, $x\in X$, and $a\in A$, $T_{\widehat P_w}^a(B\mid x) = \int_W \mathds{1}_{B}(f_a(x,v))\widehat P_w(dv)$.
    Further, define $\text{Reach}(q,a,\w^{(i)}) := \{f_a(x,\w^{(i)}) \in \reals^n : x \in q\}$.
    Then, for every $q\in Q_{\safe}$, $q'\in Q$, and $a\in A$,
    \begin{subequations}\label{eq:transition_probability_bounds}
        \begin{align}
            \min_{x\in q} T_{\widehat P_w}^a(q' \mid x) &\ge \frac{1}{N} |\{i\in\{1,\ldots,N\}:\text{Reach}(q,a, \w^{(i)})\subseteq q'\}| \label{lower:bound}\\
            \max_{x\in q} T_{\widehat P_w}^a(q' \mid x) &\le \frac{1}{N} |\{i\in\{1,\ldots,N\}:\text{Reach}(q,a, \w^{(i)})\cap q' \neq \emptyset\}|. \label{upper:bound}
        \end{align}
    \end{subequations}
\end{proposition}
\begin{proof}
    Let $x\in X$, $a\in A$. Then,
from 
the definition of the  empirical distribution, we have
%
$T_{\widehat P_w}^a(q'\mid x) 
    = \int_{W}\mathds{1}_{q'}(f_a(x,v))\widehat P_w(dv)
    = \frac{1}{N}\sum_{i=1}^N\mathds{1}_{q'}(f_a(x,\w^{(i)}))$
for all $q'\in Q$. 
For all $x\in q$, $q\in Q_{\safe}$, the previous probability is upper bounded by
$
    T_{\widehat P_w}^a(q'\mid x)\le \frac{1}{N}\sum_{i=1}^N\max_{x\in q}\mathds{1}_{q'}(f_a(x,\w^{(i)}))= \frac{1}{N} |\{i\in\{1,\ldots,N\}:\text{Reach}(q,a,\w^{(i)})\cap q' \neq \emptyset\}|,
$
%
%
which yields \eqref{upper:bound}. Analogously, we obtain the lower bound~\eqref{lower:bound},
%
%
which completes the proof.
\end{proof}
%
%
We complete the abstraction $\I$ by defining the transition probability bounds $\underline P(q,a,q'), \overline P(q,a,q')$ of $\I$ as the bounds in \eqref{eq:transition_probability_bounds} for all $q\in Q_{\safe}$, $q'\in Q$, and $a\in A$. 
\textcolor{black}{In this paper, following the approach in \cite{adams2022formal}, we over-approximate the reachable sets by using affine relaxations of $f$, which yields sound and tight transition probabilities.}
To further ensure that the paths of $\I$ that exit $Q_{\safe}$ do not satisfy the specifications, we make $q_u$ absorbing by setting $\underline P(q_u,a,q_u) = \overline P(q_u,a,q_u) = 1$ for all $a\in A$. These bounds define the set of transition probability distributions of $\I$: for all $q\in Q$, $a\in A$,
\begin{align}
\label{eq:Gamma_empirical}
    \widehat\Gamma_{q,a} = \{ \gamma\in\mathcal{P}(Q): \underline P(q,a,q')\leq \gamma(q')\leq \overline P(q,a,q')\:\text{for all}\: q'\in Q\}.
\end{align}


\subsubsection{Accounting for the Distributional Ambiguity}
\label{sec:rmdp_abstraction}

In section \ref{sec:imdp_abstraction}, we obtain an abstraction of Process~\eqref{eq:system} when the disturbance is distributed according to the center of the ambiguity set $\mathcal{D}$ in \eqref{eq:wass_ball}, i.e., $\widehat P_w$. To obtain a sound abstraction of Process~\eqref{eq:system}, we now need to account for all the distributions in said ambiguity set $\mathcal{D}$. To this end, we take $\I$ as a starting point and expand its set of transition probabilities according to $\mathcal{D}$, which yields the RMDP $\M$. However, since we consider nonlinear disturbances, we must first translate the distributional ambiguity in $\mathcal{P}(W)$ to distributional ambiguity on $\mathcal{P}(\mathbb R^n)$. Then, we leverage the results from optimal control of RMDPs \cite{gracia2023distributionally,gracia2022distributionally} to conclude correctness of the abstraction. 

Define the cost $c:Q^2\longrightarrow\reals_{\ge 0}$ such that
\begin{align*}
c(q,q') := \inf\{\|x-x'\|_l^s : x\in q, x'\in q'\}
\end{align*}
for all $q,q'\in Q$.  The sets of transition probability distributions of the RMDP abstraction are defined as follows:
\begin{definition}[Transition Probability Distributions of the RMDP Abstraction]
    \label{def:RMDP-abstraction}
    Consider the cost $c$ and the set $\widehat\Gamma$ of $\I$ as defined in \eqref{eq:Gamma_empirical}. Let $L_w(q,a) = \max_{x\in q} L_w(x,a)$ for all $q\in Q_{\safe}$, $a\in A$ and $\varepsilon(N, \beta)$ be as in Lemma~\ref{lemma:ambiguity_set} with $\beta\in (0,1)$. We define the sets of transition probability distributions of $\M$ as
    \begin{align*}
        \Gamma_{q,a} := \{\gamma\in\mathcal{P}(Q) : \exists\:\widehat\gamma \in\widehat\Gamma_{q,a} \: \text{s.t.}\: \mathcal{T}^c(\gamma,\widehat\gamma)\le (L_w(q,a)\varepsilon(\beta,N))^s\},    
    \end{align*}
    for all $q\in Q_{\safe}$, $a\in A$, and $\Gamma_{q_u,a} = \widehat\Gamma_{q_u,a}$ for all $a\in A$, so the unsafe state remains absorbing.\footnote{We can obtain a tighter representation of the uncertainty in $\Gamma_{q,a}$ by restricting ourselves to distributions in $\mathcal{P}(Q_{q,a})$, where $Q_{q,a} := \{q'\in Q : \exists x\in q,\:w\in W\:\text{s.t.}\: q' \cap Reach(x,a,w)\neq\emptyset\}$, i.e., leveraging boundedness of $W$.}
\end{definition}
The following proposition guarantees that the RMDP constructed per Def.~\ref{def:RMDP-abstraction} is a sound abstraction of Process~\eqref{eq:system}, i.e., that every $P\in\D$, when propagated through $f$, has a discrete equivalent in $\Gamma$.
%
%
\begin{proposition}[Soundness of the abstraction]
Consider the RMDP $\M = (Q,A,\Gamma, q_0, AP, L)$ as described in this section. Define, for all $x\in X$, $a\in A$, the distributions $\gamma_{x,a}\in\mathcal{P}(Q)$ such that $\gamma_{x,a}(q') := T^a(q' \mid x)$ for all $q'\in Q$. Then it holds that
\begin{align*}
    \mathbb P (
\gamma_{x,a}\in\Gamma_{q,a}\;\;\forall x\in q,\;\;\forall q\in Q,\;\;\forall a\in A ) \ge 1-\beta. 
\end{align*}
%
\end{proposition}
\begin{proof}
Let $x\in \reals^n$, $a\in A$ 
denote by $f_a(x,\cdot)_\#P_w$ the \emph{push-forward} measure of $P_w$ through $f_a(x,\cdot)$.
Start by assuming that $P_w$ and $\widehat P_w$ are $\epsilon$-close in the sense of $\mathcal{W}_s$. 
Then it holds that \cite{villani2021topics}
\begin{align*}
    \mathcal{W}_s(T^a(\cdot\mid x),T_{\widehat P_w}^a(\cdot\mid x)) = \mathcal{W}_s(f_a(x,\cdot)_\#P_w,f_a(x,\cdot)_\#\widehat P_w) \le L_w(x,a)\epsilon.
\end{align*}
Next, define the distribution $\widehat\gamma_{x,a}\in\mathcal{P}(Q)$ with $\widehat\gamma_{x,a}(q') := T_{\widehat P_w}^a(q' \mid x)$ for all $q'\in Q$. By the construction of $\widehat\Gamma_{q,a}$, $\widehat\gamma_{x,a} \in \widehat\Gamma_{q,a}$. Furthermore, from \cite[Lemma $2$]{gracia2023distributionally},
 \begin{align*}
 \mathcal{W}_s(T^a(\cdot\mid x),T_{\widehat P_w}^a(\cdot\mid x)) \le L_w(x,a)\epsilon \implies \mathcal{T}^c(\gamma_{x,a},\widehat\gamma_{x,a})^{1/s}\le L_w(q,a)\epsilon.
\end{align*}
%
By Lemma~\ref{lemma:ambiguity_set}, when $\epsilon \equiv \varepsilon(N, \beta)$, this conclusion holds with probability $1-\beta$ for all $x\in q$, $q\in Q$, $a\in A$. This in turn implies  by Definition~\ref{def:RMDP-abstraction} that also  $\gamma_{x,a} \in \Gamma_{q,a}$ for all $x\in q$, $q\in Q$, $a\in A$ with the same probability.
\end{proof}


\begin{remark}[Clustering Samples]
    \textcolor{black}{As shown in Sections~\ref{sec:learning} and \ref{sec:imdp_abstraction}, while the computational complexity of the abstraction is proportional to the number of samples, the ambiguity in the learned distribution and therefore the conservatism of the abstraction decrease with $N$. A straightforward way to deal with this trade-off}
    is to reduce the number of samples used to construct the abstraction via clustering. By doing so, we obtain another discrete distribution $\widetilde P_w$, supported on $N_{cluster} \ll N$ points, and which is very close (in the Wasserstein sense) to $\widehat P_w$. By adding this small discrepancy to the original ambiguity radius and using the triangle inequality, we construct an ambiguity ball around $\widetilde P_w$ that retains the exact same guarantees of containing $P_w$ as our original ambiguity set. 
\end{remark}

\begin{remark}
    Notice that the confidence on $\M$ being a sound abstraction of Process~\eqref{eq:system} is the same as that of our ambiguity set containing $P_w$. This is an advantage of using an ambiguity set to represent learning uncertainty, instead of obtaining bounds on individual transition probabilities, as is commonly done in the literature, e.g., \cite{badings2023probabilities, badings2023robust, ashok2019pac}. In fact, these approaches lead the resulting confidence bounds to decrease with the abstraction size and to the paradox that large abstractions can be more uncertain than coarser ones.
\end{remark}

%

\section{Strategy Synthesis}
\label{sec:strategy_synthesis}

To obtain a strategy for Process \eqref{eq:system},  we compute a strategy $\sigma^*$ for the abstraction $\M$ that maximizes the probability of satisfying $\varphi$ and is robust against all uncertainties in the abstraction. Since $\M$ is a sound abstraction of \eqref{eq:system}, $\sigma^*$ can be refined to a strategy on \eqref{eq:system} with correctness guarantees. 
To synthesize such $\sigma^*$,  we first construct a deterministic finite automaton (DFA) that represents formula $\varphi$ \cite{de2013linear}.

\begin{definition}[DFA]
    Given an \LTLf formula $\varphi$ defined over a set of atomic propositions $AP$, a \emph{deterministic finite automaton} (DFA) constructed from $\varphi$ is a tuple $\A = (Z,2^{AP},\delta,z_0,Z_F)$ where
    \begin{itemize}
        \item $Z$ is a finite set of states,
        \item $2^{AP}$ is a finite set of input symbols,
        \item $\delta:Z\times 2^{AP} \longrightarrow Z$ is the transition function,
        \item $z_0\in Z$ is the initial state,
        \item $Z_F\subseteq Z$ is the set of accepting states.
    \end{itemize}
\end{definition}

A trace $\rho = \rho_0\rho_1 \dots \rho_{K} \in (2^{AP})^*$ induces a run $z = z_0z_1 \dots z_{K+1}$ on $\A$, where $z_{k+1} = \delta(z_k,\rho_{k})$ for all $k\in \{0,\dots,K\}$. Per construction of $\A$ introduced by \cite{de2013linear}, trace $\rho$ satisfies $\varphi$ iff $z_{K+1} \in Z_F$, in which case run $z$ is called \emph{accepting} for $\A$.
%
%
Next, we generate the product of RMDP abstraction $\M$ and $\A$ to capture the paths of $\M$ that induce accepting runs in $\A$.

\begin{definition}[Product RMDP]
    Given RMDP $\M$ and DFA $\A$, the product $\PRMDP = \M\otimes\A$ is another RMDP $\PRMDP = (Q^\varphi,A^\varphi,\Gamma^\varphi,q_0^\varphi, Q_{F}^\varphi)$, where
   \begin{itemize}
       \item $Q^\varphi = Q\times Z$ is the set of states,
       \item $A^\varphi = A$ is the set of actions, and $A((q,z))$ denotes the actions available at state $(q,z)\in Q^\varphi$,
       \item $\Gamma^\varphi = \{\Gamma_{(q,z),a}^\varphi: {(q,z)\in Q^\varphi, a\in A^\varphi((q,z))}$  is the set of
        transition probability distributions for state-action pair $((q,z),a) \in Q^\varphi \times A((q,z))$, with
       \begin{align*}
            \Gamma_{(q,z),a}^\varphi := \{\gamma^\varphi \in \D(Q^\varphi) : \exists \gamma\in\Gamma_{q,a}\: \text{s.t.}\: \forall q'\in Q, \gamma^\varphi((q',z')) = \gamma(q') \: \text{iff}\: z' = \delta(z,L(q'))\},
       \end{align*}
       \item $q_0^\varphi = \delta(z_0,L(q_0))$ is the initial state,
       \item $Q_F^\varphi = Q\times Z_F$ is the set of accepting states,
   \end{itemize}
\end{definition}

By construction, a path of $\PRMDP$ is accepting
if and only if its projection onto $\M$ is a path $\path$ whose trace $\rho$ satisfies $\varphi$, i.e., $\rho$ is accepting by $\A$. 
Therefore, obtaining a strategy $\policy^*$ that robustly maximizes the probability of $\M$ satisfying $\varphi$ boils down to solving a \emph{robust maximal reachability probability problem} in $\PRMDP$. To this end, since $\Gamma$ and, consequently, $\Gamma^\varphi$, are built extending  the independent intervals of probabilities in the empirical IMDP with ambiguity sets of distributions in the sense of $\mathcal{T}^{c}$, we can rely on the efficient \emph{robust dynamic programming} algorithm for unbounded, i.e., infinite time horizon, reachability in \cite{gracia2022distributionally}. This one involves linear programming, and enjoys polynomial and exponential time complexity in the size of $\M$ and $\varphi$, respectively. The resulting strategy $\sigma^*_\varphi$ of $\PRMDP$
can be mapped to a memory-dependent strategy $\sigma^*$ of $\M$, where the memory corresponds to the current state in $\A$. Besides $\sigma^*$, the synthesis process also yields bounds on the probability of $\M$ satisfying $\varphi$ under strategy $\sigma^*$ from each $q\in Q$, i.e., \; $\underline p(q) := \min_{\xi \in \Gamma} Pr_\xi^{q,\policy^*}[\pathrmdp \models \varphi]$ and
$\overline p(q) := \max_{\xi \in \Gamma} Pr_\xi^{q,\policy^*}[\pathrmdp \models \varphi]$.

\subsection{Correctness}

In this section, we refine the strategy $\policy^*$ to a strategy $\policy_x$ of Process~\eqref{eq:system}, ensuring its correctness. First, let $J:\reals^n\rightarrow Q$ be the function that maps each continuous state to the corresponding region, i.e., for all $q\in Q$, $J(x) := q$ iff $x\in q$. 
Given 
$\pathX = x_0 \xrightarrow{a_0}  \ldots \xrightarrow{a_{K-1}} x_K$, we let $J(\pathX) := J(x_0) \xrightarrow{a_0} \ldots \xrightarrow{a_{K-1}} J(x_K)$ denote the corresponding path of $\M$. A switching strategy of \eqref{eq:system} is obtained as $\policy_x(\pathX) : = \policy^*(J(\pathX))$ 
\;
$\forall \pathX \in \PathX$. 
The following theorem ensures correctness of $\policy_x$.
\begin{thm}[Correctness]
    Let $\M$ be a sound RMDP abstraction of Process~\eqref{eq:system} with confidence $1-\beta$ and $\varphi$ be an \LTLf formula. Let $\policy^*$ be an optimal robust strategy of $\M$ w.r.t. $\varphi$ and $\policy_x$ the corresponding switching strategy of \eqref{eq:system}. Then it holds that, with confidence $1-\beta$,
  \begin{align*}
        P_{x_0}^{\sigma_x}[\pathX \models \varphi \land \Globally \neg \Prop_u] \in [\underline p(q), \overline p(q)],
  \end{align*}
    for all $x_0\in q$, $q\in Q$, where $\underline p(q)$ and $\overline p(q)$ are the bounds on the probability of $\M$ satisfying $\varphi$.
\end{thm}
\begin{proof}
    If $\M$ is a sound abstraction of \eqref{eq:system}, then the theorem follows by making use of Theorem $2$ in \cite{jackson2021formal} and letting the path length grow unbounded. Since the soundness argument holds with confidence $1-\beta$, $\policy_x$ is a correct strategy for \eqref{eq:system} with the same confidence.
\end{proof}
\section{Case Studies}
\label{sec:case_studies}

We consider various stochastic systems from the literature and use our framework to synthesize optimal strategies against various \LTLf specifications.
All experiments were run on an Intel Core i7 3.6GHz CPU with 32GB RAM. We consider $\beta=10^{-9}$, $l=\infty$, and $s = 1$ in all experiments. 

\begin{figure}[b]
    \begin{minipage}[b]{.24\linewidth}
        \centering
        \includegraphics[width=\linewidth]{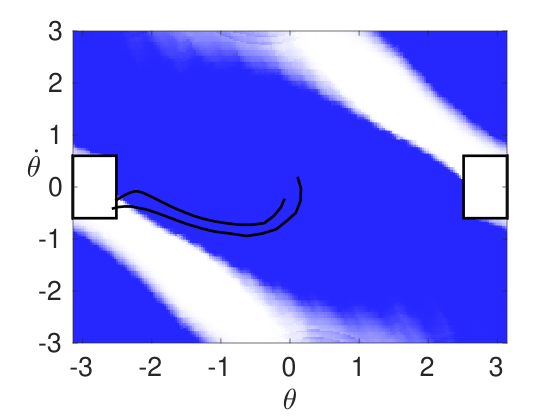}
        \subcaption{$N = 10^4$}
    \end{minipage}
    \hfill
    \begin{minipage}[b]{.24\linewidth}
    \centering
        \includegraphics[width=\linewidth]{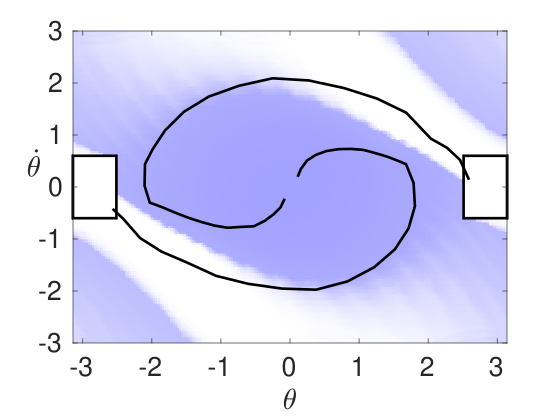}
        \subcaption{$N = 8\times10^4$}
    \end{minipage}
    \hfill
    \begin{minipage}[b]{.24\linewidth}
    \centering
        \includegraphics[width=\linewidth]{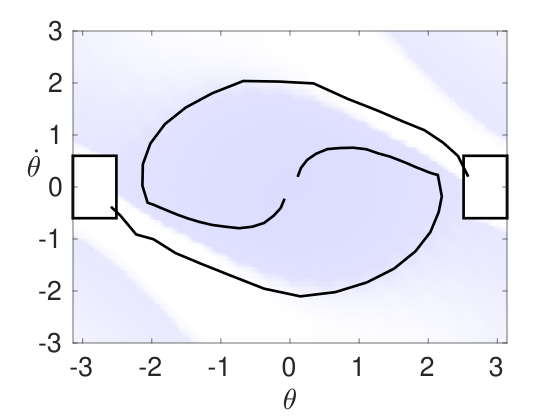}
        \subcaption{$N = 10^5$}
    \end{minipage}
    \hfill
    \begin{minipage}[b]{.24\linewidth}
    \centering
        \includegraphics[width=\linewidth]{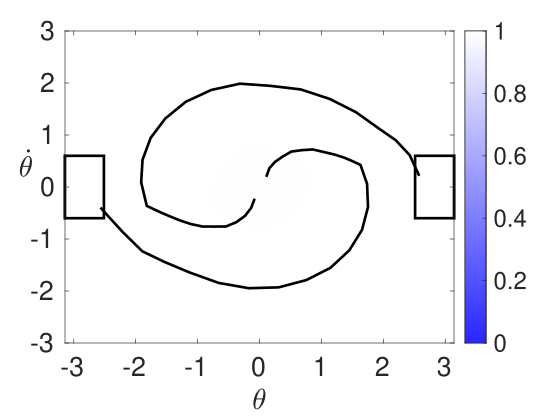}
        \subcaption{$N = 10^6$}
    \end{minipage}
    \caption{ Lower bound on the probability of satisfying $\varphi_1$ for the pendulum system as a function of $N$, together with sampled trajectories. Boxes represent the goal regions.}
    \label{fig:pendulum}
\end{figure}
\begin{table}[ht]
    \centering
    \caption{Results of the case studies. We denote by $e_{avg}$ the average difference between the lower and upper bounds in the satisfaction probabilities. Abstraction and synthesis times are given in minutes, for which respective timeouts of $8$ and $2$ hours are fixed. We denote by $\varepsilon$ and $\varepsilon_{cluster}$ the radius of the ambiguity set before and after clustering, respectively.
    }
    \scalebox{0.75}{
    \begin{tabular}{c|c|c|c|c|c|c|c|c|c}
        System (Spec.) & $|Q|$ & $|A|$ & $N$ & $\varepsilon$ & $N_\text{cluster}$ & $\varepsilon_\text{cluster}$ & $e_{avg}$ & Abstr. Time & Synth. Time \\
        \midrule
        Pendulum ($\varphi_1$) & $4\times10^4$ & $5$ & $10^{4}$ & $0.177$ & $43$ & $0.186$ & $0.744$ & $5.484$  & $7.474$ \\
         &  &  &  & $0.177$ & $-$ & $-$ & $-$ & timeout  & $-$ \\
         &  &  & $8\times10^4$ & $0.063$ & $48$ & $0.071$ & $0.219$ & $5.438$  & $120.000$ \\
         &  &  & $10^{5}$ & $0.056$ & $47$ & $0.065$ & $0.076$ & $5.406$ & $120.000$ \\
         &  &  & $10^{6}$ & $0.018$ & $49$ & $0.027$ & $0$ & $5.273$  & $61.167$ \\
        \midrule
        $3D$ Unicycle ($\varphi_1$) & $6.4\times10^4$ & $10$ & $10^{7}$ & $0.0014$ & $393$ & $0.00245$ & $0.517$ & $22.632$ & $51.942$\\
        & & & & $0.0014$ & $-$ & $-$ & $-$ & timeout & $-$\\
        & & & $10^{8}$ & $0.0047$ & $453$ & $0.0161$ & $0.498$ & $25.438$ & $43.183$\\
        & & & $5\times10^{8}$ & $0.0022$ & $4505$ & $0.006$ & $0.453$ & $235.950$ & $46.117$\\
        & & & & $0.0022$ & $8869$ & $0.0048$ & $0.447$ & $457.431$ & $43.342$\\
        \midrule
        Multiplicative & $10^{4}$ & $1$ & $4.74\times10^{3}$ & $0.034$ & $119$ & $0.041$ & $0.452$ & $1.598$ & $4.659$ \\
        noise ($\varphi_1$) &  &  &  & $0.034$ & $-$ & $-$ & $0.440$ & $59.947$ & $4.691$ \\
         \cite{skovbekk2023formal} &  &  & $4.67\times10^{4}$ & $0.012$ & $262$ & $0.017$ & $0.387$ & $3.241$ & $4.036$ \\
         &  &  &  & $0.012$ & $-$ & $-$ & $-$ & timeout & $-$ \\
         & & & $4.66\times10^{5}$ & $0.004$ & $1066$ & $0.0065$ & $0.323$ & $13.149$ & $3.863$ \\
         \midrule
         $2D$ Unicycle ($\varphi_1$) &$1.6\times10^3$ & $8$& $4\times10^5$ & $0.009$ & $385$ & $0.01$ &  $0.334$ & $0.224$ & $3.598$ \\
         \cite{gracia2023distributionally} & & &  & $0.009$ & $-$ & $-$ & $0.284$ & $341.463$ & $3.978$ \\
         & & & $2\times10^6$ & $0.013$ & $162$ & $0.015$ &  $0.180$ & $0.438$ & $3.582$ \\
         & & &  & $0.015$ & $-$ & $-$ &  $-$ & timeout  & $-$ \\
         & & & $5\times10^7$ & $0.004$ & $741$ & $0.005$ & $0.051$ & $0.735$ & $3.208$ \\
         \midrule
         $2D$ Unicycle ($\varphi_2$) &$3.6\times10^3$ & $8$& $10^{4}$ & $0.283$ & $16$ & $0.292$ & $0.494$ & $0.049$ & $5.538$ \\
          & & &  & $0.283$ & $-$ & $-$ & $0.484$ & $19.567$ &  $5.260$\\
         & & & $10^{5}$ & $0.100$ & $17$ & $0.109$ & $0.210$ & $0.044$ & $5.423$ \\
         & & &  & $0.100$ & $-$ & $-$ & $0.195$ & $180.454$ & $4.072$ \\
         & & & $10^{6}$ & $0.035$ & $44$ & $0.039$ & $0.079$ & $0.1021$ & $3.470$\\
         & & &  & $0.035$ & $-$ & $-$ & $-$ & timeout & $-$ \\
         & & & $10^{7}$ & $0.012$ & $46$ & $0.016$ & $0.03$  & $0.106$ & $3.043$ \\
    \end{tabular}
    }
    \label{tab:case_studies}
\end{table} 
In Table \ref{tab:case_studies} we report our results on $3$ nonlinear and $2$ linear systems. These models include a nonlinear pendulum in the presence of random wind and where the aerodynamic torque is proportional to $-(\dot\theta_t - \w_t\cos\theta_t) | \dot\theta_t - \w_t\cos\theta_t |$, and a $3D$ car model \cite{rajamani2011vehicle} where (nonlinear) \emph{Coulomb} friction is modeled by adding a random perturbation to the linear velocity. We also reproduce the results in \cite{skovbekk2023formal} for a system with multiplicative noise, and the results in \cite{gracia2023distributionally} for a $2D$ unicycle with additive noise. The specification considered for these systems is a reach-avoid one, i.e., $\varphi_1 = \Globally(\rm{safe})\land\Eventually(\rm{goal})$. To show that our approach can also tackle more complex formulas, we consider a $2D$ unicycle under specification $\varphi_2 := \Globally(\rm{safe}) \land \Globally(\rm{water}\rightarrow(\neg \rm{charge}\: \mathcal{U}\rm{carpet})))\land\Eventually(\rm{charge})$ \cite{vazquez2018learning}, which represents the task of simultaneously reaching a charge station while remaining safe and, if the system goes through a region with water, then it should first dry in a carpet before charging.

The results illustrate how our framework is able to provide non-trivial certificates of correctness and synthesize controllers for each of the examples. As expected from the theoretical results in Section~\ref{sec:learning}, it is possible to observe how the precision of our approach increases with the number of samples, i.e., $e_{avg}$, the difference between upper and lower bound in the probability of satisfying the \LTLf property, decreases with $N$. This can also be observed in Figure~\ref{fig:pendulum}, where we show the lower bound on the satisfaction probability for the pendulum system, for different values of $N$, while keeping all the other parameters constant.   Similarly, $e_{avg}$ decreases by increasing $|Q|$, the size of the abstraction. As a trade-off, the computation times increases with both $|Q|$ and $N$. This highlights the importance of clustering the samples, which we can observe in Table \ref{tab:case_studies}  to only have a relatively small effect in conservatism, while leading to substantial improvements in the time to build the abstraction. Finally, we empirically validated the theoretical bounds in the satisfaction probabilities via thousands of Monte Carlo simulations for random initial conditions. The empirical satisfaction probabilities lied within the theoretical bounds in all cases. 
\section{Conclusion and Future Work}

We present a framework to perform formal control of switched stochastic systems with general nonlinear dynamics and unknown disturbances under \LTLf specifications. The experimental results show the generality and effectiveness of our approach. Future research directions include increasing the tightness of our results and analyzing convergence of our solution to 
the optimal one.

\section{Acknowledgements}

This work is supported in part by National Science Foundation (NSF) under grant number 2039062 and Air Force Research Lab (AFRL) under agreement number FA9453-22-2-0050.

\bibliographystyle{IEEEtran}\bibliography{refs}

\end{document}